\documentclass[twoside,twocolumn,english,leqno]{article}
\usepackage[T1]{fontenc}
\usepackage[latin9]{inputenc}
\usepackage{graphicx}
\usepackage{esint}

\makeatletter
\usepackage{soda2e}

\makeatother

\usepackage{babel}

\begin{document}

\title{{\Large Mechanism Design via Correlation Gap}}

\author{Qiqi Yan%
\thanks{This research was supported by a Stanford Graduate Fellowship. Part
of this research was done while the author was at Yahoo! Research,
Santa Clara.%
}\\
 qiqiyan@cs.stanford.edu\\
 Department of Computer Science\\
 Stanford University}

\maketitle
%\pagenumbering{arabic}
%\setcounter{page}{1}%Leave this line commented out.

\begin{abstract}
For revenue and welfare maximization in single-dimensional Bayesian
settings, Chawla et al.~(STOC10) recently showed that sequential
posted-price mechanisms (SPMs), though simple in form, can perform
surprisingly well compared to the optimal mechanisms. In this paper,
we give a theoretical explanation of this fact, based on a connection
to the notion of correlation gap.

Loosely speaking, for auction environments with matroid constraints,
we can relate the performance of a mechanism to the expectation of
a monotone submodular function over a random set. This random set
corresponds to the winner set for the optimal mechanism, which is
highly correlated, and corresponds to certain demand set for SPMs,
which is independent. The notion of correlation gap of Agrawal et
al.\ (SODA10) quantifies how much we {}``lose'' in the expectation
of the function by ignoring correlation in the random set, and hence
bounds our loss in using certain SPM instead of the optimal mechanism.
Furthermore, the correlation gap of a monotone and submodular function
is known to be small, and it follows that certain SPM can approximate
the optimal mechanism by a good constant factor.

Exploiting this connection, we give tight analysis of a greedy-based
SPM of Chawla et al.\ for several environments. In particular, we
show that it gives an $e/(e-1)$-approximation for matroid environments,
gives asymptotically a $1/(1-1/\sqrt{2\pi k})$-approximation for
the important sub-case of $k$-unit auctions, and gives a $(p+1)$-approximation
for environments with $p$-independent set system constraints.
\end{abstract}

\section{Introduction}

In mechanism design, or even more broadly in algorithm design as well,
there is an inherent conflict between optimality and simplicity. Mechanisms
like Myerson's mechanism \cite{M81} or the VCG mechanism \cite{vic-61,cla-71,gro-73}
have optimal revenue or welfare guarantees, but often suffer from
having complicated formats or severe computational overhead. For example,
even in single-item auctions, the need for the agents to commit to
the auction process itself can be a significant burden \cite{Aus06,Hol08},
and in combinatorial auctions, determining the allocation and payments
of the VCG mechanism is a computationally hard problem \cite{NR-00}.
Therefore, not surprisingly, simple mechanisms are very often favored
in practice \cite{Hol08}. Consider sequential posted-price mechanisms,
in which the seller makes take-it-or-leave-it price offers to agents
one by one. Such mechanisms are easy to run for the sellers, leave
little room for agents' strategic behavior, and keep the information
elicitation from the agents at a minimum level. Of course, simplicity
comes at a cost, as such simple mechanisms are in general not optimal.
Therefore, it is an interesting question to quantify how much we are
paying for keeping it simple.

Following a recent trend \cite{HR09,DRY10,CHMS10}, we focus on quantifying
the performance (revenue or welfare) of simple mechanisms relative
to that of the optimal mechanism. In particular, we focus on revenue
and welfare maximization in single-dimensional Bayesian mechanism
design, and we are interested in comparing the performance of Sequential
Posted-price Mechanisms (SPMs) to that of the optimal mechanism, which
is Myerson's mechanism for revenue, and the VCG mechanism for welfare.
In a recent work of Chawla et al.~\cite{CHMS10}, it was shown for
several contexts that the performance of a SPM (which we call greedy-SPM)
approximates that of the optimal mechanism by a small constant factor,
where the factor is $2$ for matroid environments (which generalize
$k$-unit auctions, certain matching markets etc.), and $e/(e-1)$
for $k$-unit auctions. This is surprising, as SPMs can only offer
prices to agents in a very restricted way, while the optimal mechanism
can choose a price for each agent based on full information about
all other agents. What is the underlying reason for SPMs' good performance?
Our main goal of this paper is to give a theoretical explanation for
this curious fact, based on a connection to the notion of correlation
gap.

\textbf{Reducing Mechanism Design to Correlation Gap} The notion of
correlation gap was first formalized in Agrawal et al.~\cite{ADSY10}.
Let $f(S)$ be a function that maps a subset $S$ of a finite ground
set $N$ to a nonnegative real number. For $\mathcal{D}$ a distribution
over $2^{N}$ with marginal probabilities $q_{i}=Pr_{S\sim\mathcal{D}}[i\in S]$,
let $\mathcal{I}_{\mathcal{D}}$ be the independent distribution where
each $i\in N$ is included in the set with the same marginal probability
$q_{i}$, but independently. The correlation gap of $f$ is defined
as the supremum of $\frac{E_{S\sim\mathcal{D}}[f(S)]}{E_{S\sim\mathcal{I}(\mathcal{D})}[f(S)]}$
over all distribution $\mathcal{D}$, which in some sense bounds our
{}``loss'' in expected value of the function by ignoring correlation.

Loosely speaking, the approximation ratio of SPMs w.r.t.\ the optimal
mechanism is related to correlation gap in the following way. The
performance of a mechanism can often be related to the expectation
of certain function $f$ over a random set of agents. For the optimal
mechanism, this random set corresponds to the set of winners, while
for an SPM, this random set corresponds to the demand set, which is
the set of agents whose values beat the prices set for them in the
SPM. Notice that the winner set is highly-dependent, while the demand
set is independent. By setting prices for agents in an SPM carefully
such that these two random sets have the same marginal probabilities,
we can apply the correlation gap of $f$ to get a bound on the approximation
ratio of the SPM w.r.t.\ the optimal mechanism.

\textbf{Reduction for $k$-Unit Auctions} To illuminate the idea,
suppose we sell $k$ items to a set of $n$ agents $N=\{1,\dots,n\}$
with valuations drawn i.i.d.\ from a normal distribution $F$, and
our goal is to maximize expected revenue. Define set function $f$
as $f(S)=\min(|S|,k)$ for $S\subseteq N$. Let $q$ be the probability
that Myerson's optimal mechanism sells to a particular agent (which
is the same for every agent by symmetry). It can be shown that the
optimal way to sell to an agent with success probability $q$ in an
incentive compatible manner is to offer the deterministic price $p=F^{-1}(1-q)$.
Therefore if we pretend that an agent pays $p$ whenever she wins
in the optimal mechanism, the total calculated revenue is only higher.
In other words, the revenue of Myerson's mechanism is upper-bounded
by $E_{W}[f(W)]\cdot p$, where $W$ is the set of winners. On the
other hand, let an SPM make take-it-or-leave-it offers at price $p$
to every agent sequentially. Define demand set $D$ as the set of
agents whose values are at least $p$. Since at most $k$ agents can
be served, the revenue of the SPM is equal to $E_{D}[f(D)]\cdot p$.
Note that $W$ and $D$ have the same marginal probability $q$ for
every $i$, and $D$ follows an independent distribution. Therefore
if we can show that the correlation gap of $f$ is at most $\beta$,
then $E_{D}[f(D)]\geq(1/\beta)\cdot E_{W}[f(W)]$, and it follows
that the revenue of SPM is a $\beta$-approximation to that of Myerson's
mechanism.

\textbf{Submodularity} The set function $f$ that arises in our context
is the weighted rank function of the set system that encodes the feasibility
constraints of the environment. For settings where constraints are
modeled by matroids, the weighted rank functions are well-known to
be monotone and submodular. This fact enables us to invoke a deep
result from the literature on submodular functions \cite{Von07submod,ADSY10},
which says that the correlation gap of a monotone and submodular function
is at most $e/(e-1)$. It follows that for matriod environments, SPMs
can approximate the optimal mechanism by a factor of $e/(e-1)$. This
result would be otherwise difficult to achieve without making use
of our explicit connection to correlation gap and submodularity.

Recognizing submodularity is also helpful in other ways. In the analysis
for $k$-unit auctions, we exploit the cross-convexity of the multi-linear
extension of submodular functions to get a tight bound on the correlation
gap of the corresponding weighted rank function.

\textbf{Applying the Reduction} The reduction to correlation gap gives
us a structured way of analyzing greedy-SPM. It abstracts away all
the mechanism design aspects of the problem, such that we can focus
on the purely mathematical question of quantifying correlation gaps
of weighted rank functions. Based on this approach, we give tight
analysis for greedy-SPM in several contexts. In the following, approximation
guarantees are for an objective that can be revenue or welfare or
certain combination of both, and are for the version of greedy-SPM
that is tailored to the objective.

For matroid environments, as mentioned above, we show that greedy-SPM
is a $e/(e-1)$-approximation to the optimal mechanism, an improvement
over the previous $2$-approximation. For the important sub-case of
$k$-unit auctions, we show that greedy-SPM has approximation ratio
$1/(1-\frac{k^{k}}{e^{k}k!})$. ($\approx1/(1-\frac{1}{\sqrt{2\pi k}})$
by Stirling's formula) This implies that the performance of SPMs can
approach that of the optimal mechanism as the supply increases. In
particular here we do not assume that agents' valuation distributions
are identical. Towards settings more general than matroid environments,
we study $p$-independent environments, where feasibility constraints
are modeled by $p$-independent set systems, a generalization of intersection
of $p$ matroids. In such settings, we show that correlation gap is
at most $p+1$, which also translates into a $(p+1)$-approximation
for greedy-SPM. This generalizes the result on intersection of $p$
matroids in \cite{CHMS10}.

\subsection{Related Work}

For the maximization of revenue and welfare, Myerson's mechanism \cite{M81,BR-89}
and the VCG mechanism \cite{vic-61,cla-71,gro-73} are optimal, respectively.
Recent work in the CS literature has focused on designing simple mechanisms
that are approximately-optimal, while being more detail-free or robust.
Myerson's mechanism in general involves the calculation of (ironed)
virtual valuations using full distribution information. Hartline and
Roughgarden \cite{HR09} showed that a simpler mechanism, namely the
VCG mechanism with monopoly reserves is approximately optimal by a
constant factor for many natural settings, and Dhangwatnotai et al.~\cite{DRY10}
further removed the need of knowing monopoly reserves in advance via
a sampling-based approach. In another direction, Sundararajan and
Yan \cite{SY10} studied mechanisms that are approximately optimal
for utility-maximizing risk-averse sellers, even without prior knowledge
about their concave utility functions.

Sequential posted-price mechanisms have also been a recent focus of
study due to their simplicity and various appealing properties. Blumrosen
and Holenstein \cite{BH08} first compared SPMs to Myerson's mechanism
for single-item auctions by an asymptotic analysis. Chawla et al.~\cite{CHMS10}
studied SPMs in various auction contexts, proving that SPMs perform
very well compared to Myerson's mechanism, which motivated our work.
They also used SPMs as a building block to construct approximately-optimal
mechanisms in multi-dimensional settings. Independent of our work,
Chakraborty et al.~\cite{CEGMM10} proved almost the same approximation
guarantee for $k$-unit auctions. They also studied SPMs that adaptively
choose prices and the ordering of agents. Babaioff et al.~\cite{BBDS11}
studied adaptive SPMs in settings where agents' valuations are drawn
i.i.d.\ from an unknown distribution. In other aspects, Sundararajan
and Yan \cite{SY10} studied the performance of SPMs when the sellers
are risk-averse, and aim to maximize expected utility.

There is a vast literature on the study of submodular functions (see
references in \cite{Von07submod}). The correlation gap of monotone
submodular functions was first bounded in \cite{CCPV07}, and it is
also tightly related to the submodular welfare maximization problem
\cite{Von08}. In the context of auctions, Dughmi et al.~\cite{DRS09}
showed that in matroids environments, the revenue of Myerson's mechanism
is submodular in the set of agents that we actually run the mechanism
over.

\section{\label{s:setting}Preliminaries}

\textbf{Auction\ Environments} In our setting, the seller sells services
(or goods) to a set of $n$ unit-demand agents $N=\{1,\dots,n\}$.
Each agent $i$ has a private valuation $v_{i}$ for winning the service,
and $0$ otherwise, where each $v_{i}$ is drawn independently from
a known distribution $F_{i}$. For simplicity we assume that every
distribution is over a finite support $[0,L]$ for some large $L$,
and has a positive smooth density function. It is only feasible for
the seller to service certain subsets of the agents simultaneously,
and we let $\emptyset\in\mathcal{I}\subseteq2^{N}$ represent all
the feasible subsets. We assume that the environment is always downward-closed,
in the sense that the subset of a feasible set is also feasible. Auction
environments are classified by the set systems $(N,\mathcal{I})$.
In particular we study $k$-unit auctions, where a set $S$ is in
$\mathcal{I}$ if and only if $|S|\leq k$, matroid environments,
where $(N,\mathcal{I})$ forms a matroid, and $p$-independent environments,
where $(N,\mathcal{I})$ form a $p$-independent set system. We will
define the latter two environments later.

\textbf{Mechanisms} A (deterministic) mechanism uses an allocation
rule $\mathbf{x}:[0,\infty)^{n}\to\{0,1\}^{n}$ to choose the (characteristic
vector of) winning set of agents based on the reported valuations
$\mathbf{v}\in[0,\infty)^{n}$ of the agents, and uses a payment rule
$\mathbf{p}:[0,\infty)^{n}\to[0,\infty)^{n}$ to charge payments from
the agents. A randomized mechanism is a distribution over deterministic
mechanisms. For ease of presentation, we study mechanisms that are
incentive compatible (a.k.a., truthful) and individual rational, both
in the\emph{ ex post} sense, although our results still hold if we
allow mechanisms to be Bayesian incentive compatible. An equivalent
way of defining \emph{ex post} incentive constraints is that for each
agent $i$, if we fix the valuations \textbf{$\mathbf{v}_{-i}$ }of
the other agents, agent $i$ faces a take-it-or-leave-it offer at
a price $p_{i}(\mathbf{v}_{-i})$ that is independent of agent $i$'s
own value $v_{i}$.

Given an ordering of agents and a price $p_{i}$ for each agent $i$,
a Sequential Posted-price Mechanism (SPM) first initializes the allocated
set $A$ to be $\emptyset$, and for all agents $i$ in the given
order, do the following: if serving $i$ is feasible, i.e., $A+i\in\mathcal{I}$,
offer to serve agent $i$ at the pre-determined price $p_{i}$, and
add $i$ to $A$ if agent $i$ accepts. A randomized SPM is then a
distribution over deterministic SPMs.

\textbf{Weighted Rank Functions}For a set system $(N,\mathcal{I})$
with nonnegative weights $(w_{i})_{i\in N}$ on the elements, we define
the weighted rank function $w^{*}(S)$ as the maximum of $\sum_{i\in T}w_{i}$
over all $T\subseteq S$ with $T\in\mathcal{I}$. The (unweighted)
rank functions are defined with weights set to 1.

\textbf{Greedy} Given a set system $(N,\mathcal{I})$ with nonnegative
weights $(w_{i})_{i\in N}$, and a subset $S$ of $N$, the greedy
algorithm starts with an empty solution set $A$, and for each agent
$i$ in $S$ in decreasing order of $w_{i}$, adds $i$ into the solution
set $A$ whenever $A\cup\{i\}$ is in $\mathcal{I}$. Finally it outputs
$A$. We let $greedy(S)$ denote the final output of greedy algorithm. 

\textbf{Matroids }A set system $(N,\mathcal{I})$ is a matroid system
(see e.g.\ \cite{Ox-92}) if (1) $S\in\mathcal{I}$ whenever $S\subseteq T\in\mathcal{I}$,
and (2) if $S,T\in\mathcal{I}$ and $|S|>|T|$, then for some $e\in S\backslash T$,
$T\cup\{e\}\in\mathcal{I}$. We will make use of the following two
well-known properties about matroids: (1) If we run greedy on a subset
$S$ of the matroid, then the weight of its output set equals to the
weighted rank of the set $S$, i.e., $\sum_{i\in greedy(S)}w_{i}=w^{*}(S)$.
(2) The weighted rank function of a matroid is monotone and submodular.

\textbf{Correlation Gap and Submodularity} Given a set function $f:2^{N}\to[0,\infty)$
over a finite set $N$, let $\mathcal{D}$ be a distribution over
$2^{N}$ with marginal probabilities $\mathbf{q}=(q_{i})_{i\in N}$.
Let $S\sim\mathcal{I}(\mathcal{D})$ denote that each $i\in N$ is
included in $S$ with probability $q_{i}$ independently. Then the
correlation gap \cite{ADSY10} \footnote{We differ from \cite{ADSY10} in that the correlation gap was defined there to be at most 1.}
of $f$ is $\sup_{\mathcal{D}}\frac{E_{S\sim\mathcal{D}}[f(S)]}{E_{S\sim\mathcal{I}(\mathcal{D})}[f(S)]}$.
(we let $\frac{0}{0}=1$ here)

A set function $f:2^{N}\to[0,\infty)$ is monotone if $f(S)\leq f(T)$
whenever $S\subseteq T$, and is submodular if $f(S)+f(T)\geq f(S\cup T)+f(S\cap T)$
for all $S,T$.

\begin{theorem} \label{t:cg_submod}\cite{CCPV07,ADSY10} The correlation
gap of a monotone submodular function is at most $e/(e-1)$. 

\end{theorem}

\section{\label{sec:reduction}Posted-Price\emph{ vs }Optimal: A Reduction
to Correlation Gap}

We will focus on comparing SPMs with the optimal mechanism in the
context of revenue maximization. Almost identical claims can be made
for welfare and certain other objectives, which we discuss in Section
\ref{sec:welfare}.

\subsection{A Single Bidder Optimization Problem}

Before we embark on studying mechanisms that involve multiple bidders,
it is crucial to first understand the following optimization problem
that involves only one bidder.

\begin{Problem}Given an agent with valuation distribution $F$, and
a target selling probability $0<q<1$, what price distribution $\mathcal{D}$
maximizes our expected revenue, i.e., $E_{p\sim\mathcal{D}}[p\cdot(1-F(p))]$,
subject to the constraint that the selling probability is exactly
$q$, i.e., $E_{p\sim\mathcal{D}}[1-F(p)]=q$. 

\end{Problem}

To study this problem, first suppose that we can only offer a deterministic
price. Then for any selling probability $q$, our only choice is to
offer the deterministic price $F^{-1}(1-q)$, and the revenue we get
as a function of $q$ is $R_{F}(q)=q\cdot F^{-1}(1-q)$.

Now suppose instead we are allowed to offer a random price, then we
can do possibly better. To be specific, we can randomize between two
prices $\underline{p}$ and $\bar{p}$ with selling probabilities
$\underline{q}=1-F(\underline{p})$ and $\bar{q}=1-F(\bar{p})$ satisfying
$\underline{q}\leq q\leq\overline{q}$, and in particular we draw
$\underline{p}$ with probability $\frac{\overline{q}-q}{\overline{q}-\underline{q}}$
and draw $\overline{p}$ with probability $\frac{q-\underline{q}}{\overline{q}-\underline{q}}$
such that the selling probability is exactly equal to $q$. Then our
revenue is equal to $\frac{\overline{q}-q}{\overline{q}-\underline{q}}\cdot R_{F}(\underline{q})+\frac{q-\underline{q}}{\overline{q}-\underline{q}}\cdot R_{F}(\overline{q})$,
which is possibly better than $R_{F}(q)$. Let $\overline{R}_{F}(q)$
be the maximum revenue one can get by randomizing between two prices
this way. One can show that $\overline{R}_{F}$ equals to the concave
closure of $R_{F}$, i.e., the minimum concave function that upper-bounds
$R_{F}$. Moreover, the optimal distribution is in fact just the two-price
distribution that gives $\overline{R}_{F}(q)$.

\begin{figure}
\includegraphics[scale=1.75]{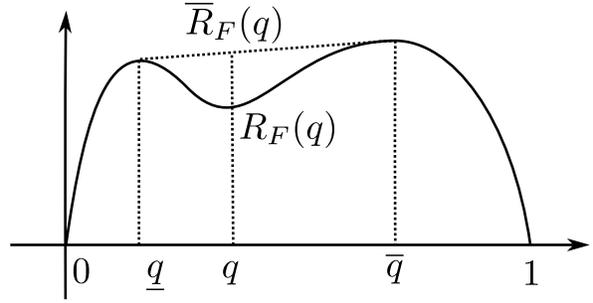}\caption{Revenue Curve and ``Ironed'' Revenue Curve}

\end{figure}

In the well-known special case that $F$ is regular, i.e., $R_{F}(q)$
is concave in $q$\footnote{This is equivalent to the definition that the virtual valuation function is monotone.},
the two-price distribution degenerates to a single deterministic price
$F^{-1}(1-q)$, and $\overline{R}_{F}(q)=R_{F}(q)$ in this case.

For the purpose of the rest of the paper, the following lemma summarizes
this discussion.

\begin{lemma}\cite{M81,BR-89} \label{lem:ironing}For all valuation
distribution $F$ and probability $q$, the price distribution $\mathcal{D}$
that maximizes $E_{p\sim\mathcal{D}}[p\cdot(1-F(p))]$ subject to
the constraint that $E_{p\sim\mathcal{D}}[1-F(p)]=q$ is a two-price
distribution, where this distribution as well as the revenue $\overline{R}_{F}(q)$
it gives us can be determined from $F$. Moreover, $\overline{R}_{F}(q)$
is a concave function.

\end{lemma}

For notational convenience, we will use $\overline{R}_{i}$ to denote
the $\overline{R}_{F}$ function for agent $i$.

\subsection{\label{sec:reduction_theorem}Reduction Theorem: the Revenue Case}

\begin{Definition} The greedy-SPM of Chawla et al.\ \cite{CHMS10}
(with slight changes) does the following:
\begin{enumerate}
\item For each agent $i$, calculate $q_{i}$, the winning probability of
agent $i$ in Myerson's mechanism. Remove agent $i$ if $q_{i}=0$.
\item For each agent $i$, draw a random price $p_{i}$ from the optimal
price distribution w.r.t.\ distribution $F_{i}$ and selling probability
$q_{i}$ according to the Ironing Lemma.
\item Let $A=\emptyset$. For all agent $i$ in decreasing order of effective
prices $\hat{p}_{i}$ defined as $\hat{p}_{i}=\overline{R}_{i}(q_{i})/q_{i}$,
if serving agent $i$ is feasible, i.e., $A+i\in\mathcal{I}$, offer
price $p_{i}$ to agent $i$, and add $i$ into $A$ if agent $i$
accepts. 
\end{enumerate}
\end{Definition} 

\begin{theorem}[Reduction Theorem for Matroids] 

\label{thm:reduction}For matroid environments, if the correlation
gap of the weighted rank function is at most $\beta$ for no matter
what non-negative weights, then the expected revenue greedy-SPM is
a $\beta$-approximation to that of Myerson's optimal mechanism.

\end{theorem}

\begin{proof} 

In the following two claims, we relate the expected revenue of both
Myerson's mechanism and greedy-SPM to the weighted rank function with
effective prices $\hat{p}_{i}$ as weights, which we denote as $\hat{p}^{*}$.

\begin{claim} 

Let $W$ be the (random) set of winning agents in Myerson's mechanism.
The expected revenue of Myerson's mechanism is upper-bounded by $E_{W}[\hat{p}^{*}(W)]$.

\end{claim}

\begin{proof} Let $q_{i}=Pr_{W}[i\in W]$ be the probability that
agent $i$ wins in Myerson's mechanism. By Lemma \ref{lem:ironing},
the optimal way to sell to agent $i$ with probability $q_{i}$ gives
expected revenue $\overline{R}_{i}(q_{i})$. By linearity of expectation,
the expected revenue of Myerson's mechanism is upper-bounded by $\sum_{i\in N}\overline{R}_{i}(q_{i})$.
To relate this to the effective prices, suppose in Myerson's mechanism,
we get effective payment $\hat{p}_{i}$ whenever agent $i$ wins.
Then the total effective revenue is $E_{W}[\sum_{i\in W}\hat{p}_{i}]$.
Also, each agent $i$ wins with probability $q_{i}$ in Myerson's
mechanism, contributing $q_{i}\hat{p}_{i}=\overline{R}_{i}(q_{i})$
to total effective revenue, and hence $\sum_{i\in N}\overline{R}_{i}(q_{i})$
equals effective revenue $E_{W}[\sum_{i\in W}\hat{p}_{i}]$. Further,
since $W$ is a feasible set, we can rewrite $E_{W}[\sum_{i\in W}\hat{p}_{i}]$
as $E_{W}[\hat{p}^{*}(W)]$, and our claim follows.

\end{proof} 

\begin{claim} Let demand set $D$ be the (random) set of agents whose
values beat the prices set for them respectively. The expected revenue
of greedy-SPM equals to $E_{D}[\hat{p}^{*}(D)]$.

\end{claim}

\begin{proof} Because valuation distributions of the agents are independent,
each agent $i$ is in the demand set $D$ with probability $q_{i}$
independently. Observe that ignoring agents not in the demand set,
who do not win anyway, greedy-SPM effectively runs the greedy algorithm
on the demand set $D$ w.r.t.\ weights $\hat{p}_{i}$ subject to
feasibility constraints. The expected effective revenue of greedy-SPM
is hence equal to $E_{D}[\sum_{i\in greedy(D)}\hat{p}_{i}]$, which
is equal to $E_{D}[\hat{p}^{*}(D)]$ by the optimality of greedy for
matroid. Note that whenever the random price $p_{i}$ is offered to
an agent, we get expected revenue $\overline{R}_{i}(q_{i})$, while
the expected effective revenue is $q_{i}\hat{p}_{i}$, also equal
to $\overline{R}_{i}(q_{i})$. Therefore the expected revenue of greedy-SPM
equals to the expected effective revenue, which is $E_{D}[\hat{p}^{*}(D)]$.

\end{proof}

By our assumption that the correlation gap of the weighted rank function
is at most $\beta$, we have $E_{D}[\hat{p}^{*}(D)]\geq\frac{1}{\beta}\cdot E_{W}[\hat{p}^{*}(W)]$,
and our theorem follows by chaining this inequality with the above
two claims.

\end{proof}

For settings beyond matroids, we need the following technical condition
for the reduction to work, which is a stronger condition than merely
a bound on correlation gap.

\begin{Definition}

We say that the greedy algorithm verifies a correlation gap of $\beta$
for the weighted rank function of a set system, if for all nonnegative
weights $(w_{i})_{i\in N}$, and distribution $\mathcal{D}$ over
$2^{N}$, we have $E_{S\sim\mathcal{I}(\mathcal{D})}[\sum_{i\in greedy(S)}w_{i}]\geq\frac{1}{\beta}E_{S\sim\mathcal{D}}[w^{*}(S)]$. 

\end{Definition} 

\begin{theorem}[Reduction Theorem in General] 

For any downward-closed environment, if the greedy algorithm verifies
a correlation gap of $\beta$ for the weighted rank function for arbitrary
non-negative weights, then the expected revenue of greedy-SPM is a
$\beta$-approximation to that of Myerson's optimal mechanism. 

\end{theorem} 

\begin{proof}

Similarly, we upper-bound the revenue of Myerson by $E_{W}[\hat{p}^{*}(W)]$,
and express the revenue of greedy-SPM as $E_{D}[\sum_{i\in greedy(D)}\hat{p}_{i}]$.
The theorem follows by applying the assumption that greedy verifies
a correlation gap of $\beta$.

\end{proof}

\begin{Remark}

One crucial property about the greedy algorithm is that although we
are running greedy on the all agents, but for \emph{no matter what}
demand set it turns out to be, greedy is also optimizing or approximately
optimizing for this demand set. Most other approximation algorithms
do not have this property.

\end{Remark}

\subsection{\label{sec:welfare}Extension to Welfare and Other Objectives}

We specify an objective by defining functions of the form $g_{i}(v,p)$
for agents. If agent $i$ has true value $v$ and is offered a price
$p$ with $v\geq p$, then agent $i$ wins, and we gain objective
value $g_{i}(v,p)$. Our goal is then to maximize the total objective
value we collect from the agents. For maximizing welfare, revenue,
and surplus, we set $g_{i}(v,p)=v$, $g_{i}(v,p)=p$, and $g_{i}(v,p)=v-p$,
respectively. One can also define other objectives this way.

To adapt the definition of greedy-SPM and our reduction theorems,
we need the following changes. We define $G_{i}(q)$ as the maximum
expected objective value the seller can get by offering a deterministic
price such that the agent wins with probability $q$. We can then
derive an Ironing Lemma similarly, and also define $\overline{G}_{i}(q)$
as the concave closure of $G_{i}(q)$. Then we use effective gain
defined as $\overline{G}_{i}(q)/q$ to replace effective prices as
weights, and the rest of the proof goes the same way.

\subsection{\label{sec:computation}Efficiently Computable SPMs}

In greedy-SPM, we need to compute the winning probabilities of the
agents in Myerson's mechanism, which is potentially computationally
hard. This was addressed in Chawla et al.~by a sampling-based approach,
which estimates the winning probabilities by repeatedly running Myerson's
mechanism for sufficiently many times.

We note that the winning probabilities give a feasible solution to
the following convex program, whose optimal value gives an upper bound
on the revenue of Myerson's mechanism.\begin{eqnarray*}
 & \mbox{maximize}\sum_{i\in N}\overline{R}_{i}(q_{i})\\
 & \mbox{subject to}\\
 & \sum_{i\in S}q_{i}\leq rank(S) & \mbox{for all }S\\
 & q_{i}\geq0 & \mbox{for all }i\end{eqnarray*}

For many settings, we can solve this convex program efficiently to
get the optimal $q_{i}$ values, and use them in greedy-SPM instead.
It turns out that for settings we study in this paper, this variant
of greedy-SPM gives the same approximation guarantees. We leave the
details of this observation to the full version of the paper.

\section{\label{sec:app}Revenue and Welfare Guarantees of Greedy-SPM}

Based on the reduction theorem, we give tight analysis of greedy-SPM
of Chawla et al., and prove the guarantees in Theorem \ref{thm:main}.
By the reduction theorem, it suffices to study the the correlation
gaps of the weighted rank functions, and the greedy algorithm, which
we do separately in the following subsections. 

\begin{theorem} 

\label{thm:main}The expected revenue of greedy-SPM is a $\beta$-approximation
to that of Myerson's optimal mechanism, and the expected welfare of
(the welfare version of) greedy-SPM is a $\beta$-approximation to
that of the VCG mechanism, where:
\begin{itemize}
\item $\beta=e/(e-1)$ for matroid environments \\
{\hfill}(an improvement over 2) 
\item $\beta=1/(1-\frac{k^{k}}{e^{k}k!})\approx1/(1-\frac{1}{\sqrt{2\pi k}})$
for $k$-unit auctions \\
{\hfill}(an improvement over $e/(e-1)$) 
\item $\beta=p+1$ for $p$-independent environments \\
{\hfill}(a generalization from intersection of $p$ matroids) 
\end{itemize}
\end{theorem} 

\begin{Remark} 

For matroid environments, as noticed in \cite{CHMS10}, if we run
the VCG mechanism, and set reserves to be the same as the prices used
in greedy-SPM, the revenue we get is as good as that of greedy-SPM,
for any particular valuation profile. It follows that the VCG mechanism
with such reserve prices has the same approximation guarantee for
revenue.

\end{Remark}

\subsection{\label{sec:matroid}Matroid Environments}

Matroid environments are important because many auction constraints
can be modeled using matroids, and matroids have various nice properties.
To give a few examples of matroids, $k$-uniform matroids encode the
constraints of $k$-unit auctions, where $S$ is in $\mathcal{I}$
if and only if $|S|\leq k$, graphical matroids enforce the feasible
sets to be the edge sets of acyclic subgraphs of a given graph, and
transversal matroids can model certain matching markets, and etc.

By the reduction theorem, to establish an $e/(e-1)$-approximation
of greedy-SPM in matroid environments, it suffices to prove the following
lemma. \begin{lemma} 

\label{l:greedy_matroids}The correlation gap of the weighted rank
function of a matroid is at most $e/(e-1)$.

\end{lemma} 

\begin{proof}

This lemma follows from the fact that the weighted rank function of
a matroid is monotone and submodular, and that the correlation gap
of a monotone submodular function is at most $e/(e-1)$.

\end{proof}

\subsection{\label{sec:k-uniform}$k$-Unit Auctions}

$k$-Unit auctions form an important sub-class of a matroid environments.
The feasibility constraints of a $k$-unit auction is modeled by a
$k$-uniform matroid. In the following, we precisely quantify the
correlation gap of the weighted rank function of $k$-uniform matroids.

For a $k$-uniform matroid over $n$ elements, the (unweighted) rank
function is $f_{n}^{k}(S)=\min(|S|,k)$ for $S\subseteq N=\{1,\dots,n\}$.
We drop superscript and subscript when the context is clear. It is
easy to verify that $f$ is monotone and submodular. Define the multi-linear
extension $Ef(\mathbf{q})$ for\textbf{ $\mathbf{q}\in[0,1]^{n}$
}(in the sense of \cite{CCPV07}) as the expectation of $f(S)$ where
each $i\in N$ is included in $S$ with probability $q_{i}$ independently.
As was shown in \cite{CCPV07}, or can be easily verified using definitions,
if $f$ is submodular, then $Ef$ satisfies cross-convexity, in the
sense that $\frac{\partial^{2}Ef(\mathbf{q})}{\partial q_{i}\partial q_{j}}\leq0$
for all $\mathbf{q}\in(0,1)^{n}$ and $i\neq j$.

For all $n$ and $0\leq k\leq n$, define $\Phi(n,k)$ as the minimum
of $Ef_{n}^{k}(\mathbf{q})$ over all marginal probability vector
$\mathbf{q}$ such that $\sum_{i\in N}q_{i}=k$. In the following
lemma, we identify the probability vector $\mathbf{q}$ that minimizes
$Ef(\mathbf{q})$ subject to this constraint, and show several useful
properties about $\Phi(n,k)$. This lemma is interesting in itself,
and in fact can be used to improve the analysis of an SPM in \cite{SY10}. 

\begin{lemma} 

\label{lem:Phi}The following holds for $\Phi(n,k)$:
\begin{description}
\item [{{(a)}}] $\Phi(n,k)=Ef_{n}^{k}(\mathbf{q)}$ where $q_{i}=k/n$
for all $i\in\{1,\dots,n\}$. In other words, $\Phi(n,k)$ is the
expected value of $\min(X,k)$, where $X$ is a binomial random variable
with parameters $n$ and $k/n$. 
\item [{{(b)}}] $\Phi(n,k)$\emph{ }monotonely\emph{ }increases\emph{
}with\emph{ $k$, }and\emph{ }monotonely\emph{ }decreases\emph{ }with\emph{
$n$.} 
\item [{{(c)}}] $\lim_{n\to\infty}\Phi(n,k)=k-\frac{k^{k+1}}{e^{k}k!}\approx k-\frac{k}{\sqrt{2\pi k}}$. 
\end{description}
\end{lemma} 

\begin{proof} 

To prove (a), first for an arbitrary marginal probability vector \textbf{$\mathbf{q}\in[0,1]^{n}$},
consider vector $\overline{\mathbf{q}}$ that is the same as $\mathbf{q}$
except that the $i$-th and $j$-th components are averaged for some
$i\neq j$, i.e., $\overline{q}_{i}=\overline{q}_{j}=(q_{i}+q_{j})/2$.
We show that $Ef(\overline{\mathbf{q}})\leq Ef(\mathbf{q})$. Let
\textbf{$\mathbf{q}'$} be the same as $\mathbf{q}$ except with the
$i$-th and $j$-th components switched, i.e., $q_{i}'=q_{j}$ and
$q_{j}'=q_{i}$. By symmetry of $f$, $Ef(\mathbf{q})=Ef(\mathbf{q}')$,
and $\overline{\mathbf{q}}$ is the middle-point of $\mathbf{q}$
and $\mathbf{q}'$. By the cross-convexity of $Ef$, the value of
$Ef$ is convex in the line segment connecting $\mathbf{q}$ and\textbf{
$\mathbf{q}'$}. Therefore $Ef(\overline{\mathbf{q}})$ is at most
the average of $Ef(\mathbf{q})$ and $Ef(\mathbf{q}')$, or simply
$Ef(\mathbf{q)}$. Now starting with an arbitrary \textbf{$\mathbf{q}$},
by repeatedly averaging the maximum and minimum components of $\mathbf{q}$
this way, the value of $Ef(\mathbf{q})$ keeps decreasing, while all
$q_{i}$'s converge to $k/n$. By the continuity of $Ef(\mathbf{q})$
in $\mathbf{q}$, the value of $Ef(\mathbf{q})$ converges to the
value of $Ef$ at $q_{i}=k/n$ for all $i$. Therefore $Ef(\mathbf{q})$
is minimized at $q_{i}=k/n$ for all $i$.

To show (b), it is obvious that $\Phi(n,k)$ is monotonely increasing
in $k$, because $f_{n}^{k}(S)$ is increasing in $k$. It suffices
to show that $\Phi(n,k)$ is monotonely decreasing in $n$. Recall
that $\Phi(n,k)$ was defined to be the optimal value of a minimization
problem. To relate $\Phi(n,k)$ to $\Phi(n+1,k)$, we cast the optimal
solution underlying $\Phi(n,k)$, which is an $n$-dimensional independent
distribution, to $(n+1)$-dimensional space, such that it gives a
candidate solution to the minimization problem underlying $\Phi(n+1,k)$.
To be specific, we observe that $\Phi(n,k)$ is equal to $Ef_{n+1}^{k}(\mathbf{q})$,
where $\mathbf{q}$ is an $(n+1)$-dimensional vector with $q_{i}=k/n$
for $i=1,\dots,n$, and $q_{n+1}=0$. By definition of $\Phi(n+1,k)$,
$\Phi(n+1,k)\geq Ef_{n+1}^{k}(\mathbf{q})=\Phi(n,k)$.

We leave the derivation of (c) to the appendix. 

\end{proof} 

Based on Lemma \ref{lem:Phi}, we can first quantify the correlation
gap of the unweighted rank function, and then extend it to the weighted
case. \begin{lemma} 

\label{lem:k-uniform}For $n,k\geq1$, the correlation gap of the
function $f(S)=\min(|S|,k)$ for $S\subseteq N=\{1,\dots,n\}$ is
exactly $\frac{k}{\Phi(k,n)}$.

\end{lemma}

\begin{proof} For any probability vector $\mathbf{q}$, let $\mathcal{O}_{\mathbf{q}}$
be the distribution over $2^{N}$ with marginal probabilities $\mathbf{q}$
that maximizes $E_{S\sim\mathcal{O}_{\mathbf{q}}}[f(S)]$. We first
show that $E_{S\sim\mathcal{O}_{\mathbf{q}}}[f(S)]$ equals $\sum_{i}q_{i}$
if $\sum_{i}q_{i}\leq k$, and equals $k$ otherwise. (1) Suppose
$\sum_{i}q_{i}\leq k$. First note that $E_{S\sim\mathcal{O}_{\mathbf{q}}}[f(S)]\leq E_{S\sim\mathcal{O}_{\mathbf{q}}}[|S|]=\sum_{i}q_{i}$.
Moreover, $\mathbf{q}$ can be seen as a point inside the integral
polytope with (characteristic vectors of) feasible sets (sets of size
at most $k$) as vertices. Then by standard polyhedral combinatorics,
one can decompose this point as a convex combination of the vertices,
which corresponds to a distribution over feasible sets with marginal
probabilities $\mathbf{q}$. This distribution gives expected $f$
value $\sum_{i}q_{i}$. (2) If $\sum_{i}q_{i}>k$, then by the monotonicity
of $E_{S\sim\mathcal{O}_{\mathbf{q}}}[f(S)]$ in $\mathbf{q}$, $E_{S\sim\mathcal{O}_{\mathbf{q}}}[f(S)]$
is at least $k$. However it is also upper-bounded by $k$ as $f$
is upper-bounded by $k$. Therefore $E_{S\sim\mathcal{O}_{\mathbf{q}}}[f(S)]=k$
in this case.

Suppose that $\mathbf{q}$ maximizes the {}``gap ratio'' $\frac{E_{S\sim\mathcal{O}_{\mathbf{q}}}[f(S)]}{E_{S\sim\mathbf{q}}[f(S)]}$.
We first show that $r=\sum_{i}q_{i}\leq k$. If this is not the case,
then by lowering the $q_{i}$'s such that $\sum_{i}q_{i}=k$, $E_{S\sim\mathbf{q}}[f(S)]$
strictly decreases, while $E_{S\sim\mathcal{O}_{\mathbf{q}}}[f(S)]$
is still $k$. This gives a strictly higher gap ratio, contrary to
that assumption that $\mathbf{q}$ maximizes the gap ratio.

Next we show that $r=k$. For $r\leq k$, we can explicitly express
the reciprocal of the gap ratio as: \begin{eqnarray*}
 &  & \frac{1}{r}\cdot\sum_{t=0}^{n}{n \choose t}\cdot\left(\frac{r}{n}\right)^{t}\cdot\left(\frac{n-r}{n}\right)^{n-t}\cdot\min(t,k)\\
 & = & \sum_{t=1}^{n}{n-1 \choose t-1}\left(\frac{r}{n}\right)^{t-1}\left(\frac{n-r}{n}\right)^{n-t}\cdot\frac{\min(t,k)}{t}\end{eqnarray*}

This is equal to the expectation of $\frac{\min(X+1,k)}{X+1}$ where
$X$ is the binomial random variable with parameters $n-1$ and $r/n$.
It is also equal to $\int_{0}^{\infty}Pr[\frac{\min(X+1,k)}{X+1}\geq x]dx$.
Note that for $x>1$, $Pr[\frac{\min(X+1,k)}{X+1}\geq x]=0$, and
otherwise $Pr[\frac{\min(X+1,k)}{X+1}\geq x]=Pr[X+1\leq k/x]$, where
$Pr[X+1\leq k/x]$ strictly decreases as $r$ increases. Therefore
the gap ratio is maximized at $r=k$.\end{proof} 

\begin{lemma} \label{lem:k-uniform-weighted}For $n,k\geq1$, the
correlation gap of the weighted rank function of a $k$-uniform matroid
of size $n$ is at most $\frac{k}{\Phi(k,n)}$.

\end{lemma}

\begin{proof} Again let $f(S)=\min(|S|,k)$ for $S\subseteq N=\{1,\dots,n\}$.
Assume w.l.o.g.\ that $w_{1}\geq w_{2}\geq\dots\geq w_{n}$, and
let $w_{n+1}=0$ for convenience. The weighted rank function $w^{*}(S)$
can be written as $\sum_{i\in N}(w_{i}-w_{i+1})\cdot f(S\cap\{1,\dots,i\})$,
a conic combination of unweighted rank functions. The correlation
gap of $w^{*}$ is therefore witnessed by the correlation gap of $f(S\cap\{1,\dots,i\})$
for some $i$, and hence it equals $\sup_{1\leq i\leq n}k/\Phi(i,k)$.
By Lemma \ref{lem:Phi}(b), $\Phi(i,k)$ is decreasing in $i$, and
hence the correlation gap of $w^{*}$ is $k/\Phi(n,k)$.

\end{proof}

\begin{Remark} We cannot generalize Lemma \ref{lem:k-uniform} or
\ref{lem:k-uniform-weighted} to work for arbitrary matroids with
rank $k$. For any $k$, consider the partition matroids with $k$
parts, each of size $n$, where a feasible set can only have at most
one element from each part. The rank of such a matroid is $k$, while
the correlation gap is the same as that of a 1-uniform matroid over
$n$ elements, which approaches $e/(e-1)$ as $n$ increases. \end{Remark}

\subsection{\label{sec:p-ind}$p$-Independent Environments}

There are interesting auction constraints that cannot be modeled by
matroids, but can be modeled by $p$-independent set systems. In a
set system $(N,\mathcal{I})$, a base of a subset $S\subseteq N$
is a maximal feasible subset of $S$. A set system $(N,\mathcal{I})$
is a $p$-independent system if for any non-empty subset $S$ of $N$:\[
\frac{\mbox{maximum size of a base of }S}{\mbox{minimum size of a base of }S}\leq p.\]

For example, a matroid is $1$-independent, and vice versa. The edge
sets of (non-bipartite) matchings of a graph form a $2$-independent
system (but in general cannot be cast as the intersection of a constant
number of matroids). The intersection of $p$ matroids is $p$-independent.
The feasible sets of agents in single-minded combinatorial auctions
with bounded bundle size $p$ form a $p$-independent system.

It is well-known that the greedy algorithm gives a $p$-approximation
for $p$-independent systems \cite{Jen76}. For our purpose, it suffices
to prove the following lemma, by combining arguments of \cite{CCPV07,CHMS10}.

\begin{lemma} 

\label{l:greedy_ind} The greedy algorithm verifies a correlation
gap of $p+1$ for $p$-independent system constraints.

\end{lemma} 

This ratio of $p+1$ is tight, up to lower order terms. 

\begin{proposition} 

\label{pro:p_ind_lb}For any sufficiently large positive integer $p$,
there is a $p$-independent set system with correlation gap at least
$p/\log p$. 

\end{proposition}

\section{Conclusion}

We summarize the main observation of this paper as follows. For revenue
and welfare maximization, the approximation ratio of certain SPM compared
to the optimal mechanism is inherently related to the correlation
gap of the weighted rank function of set system that models the feasibility
constraints. In particular for matroid environments, the weighted
rank functions have small correlation gap, which explains why SPMs
give good approximation guarantees in these settings.

Moreover, our point is made stronger by the fact that we are proving
guarantees for a very restricted type of SPMs, where prices and offering
order have to be predetermined. Our observation can be used as a guideline
for the design and analysis of more relaxed types of SPMs, which seems
to be an interesting research direction.

\textbf{Acknowledgment} This research is made possible by learning
from two {}``right'' groups of people, people behind SPMs: Shuchi
Chawla, Jason Hartline, David Malec, and Balu Sivan, and people behind
correlation gap: Shipra Agrawal and Jan Vondrák. I thank them for
sharing their intuition about the subject with me. I also thank Tim
Roughgarden, who knows both subjects, for various valuable suggestions.

\bibliographystyle{alpha} \bibliographystyle{alpha}
\bibliography{auctions}

\section{Proof of Lemma \ref{lem:Phi}}

\begin{proof} 

We derive the asymptotics for $\Phi(n,k)$ as follows, where the last
step is by Stirling's approximation of factorials. \begin{eqnarray*}
 &  & \lim_{n\to\infty}\Phi(n,k)\\
 & = & \lim_{n\to\infty}\sum_{t=0}^{n}{n \choose t}\cdot\left(\frac{k}{n}\right)^{t}\cdot\left(\frac{n-k}{n}\right)^{n-t}\cdot\min(t,k)\\
 & = & \lim_{n\to\infty}\sum_{t=0}^{k-1}{n \choose t}\cdot\left(\frac{k}{n}\right)^{t}\cdot\left(\frac{n-k}{n}\right)^{n-t}\cdot t\\
 &  & +k\cdot\left(1-\sum_{t=0}^{k-1}{n \choose t}\cdot\left(\frac{k}{n}\right)^{t}\cdot\left(\frac{n-k}{n}\right)^{n-t}\right)\\
 & = & \sum_{t=0}^{k-1}\frac{k^{t}}{t!}\cdot\frac{1}{e^{k}}\cdot t+k\cdot\left(1-\sum_{t=0}^{k-1}\frac{k^{t}}{t!}\cdot\frac{1}{e^{k}}\right)\\
 & = & k\cdot\left(1-\frac{k^{k}}{e^{k}k!}\right)\approx k\cdot\left(1-\frac{1}{\sqrt{2\pi k}}\right).\end{eqnarray*}

\end{proof}

\subsection{Proof of Lemma \ref{l:greedy_ind}}

\begin{proof} Fix marginal probabilities $\mathbf{q}$. In the dependent
case, if $S$ is drawn from a distribution $\mathcal{D}$ with marginal
probabilities\textbf{ $\mathbf{q}$}, let $\tilde{q}_{i}$ be the
probability that $i$ is in the optimal feasible subset of $S$ (with
arbitrary fixed tie-breaking). we can rewrite $E_{S\sim\mathcal{D}}[w^{*}(S)]$
as $\sum_{i\in N}\tilde{q}_{i}w_{i}$.

Now consider the independent case, where each $i$ is in $S$ with
probability $q_{i}$ independently, which we denote by $S\sim\mathbf{q}$.
Let $A=g(S)$ be the agents allocated by running the greedy algorithm
on $S$. The expected performance of greedy is $E_{S\sim\mathcal{D}}[\sum_{i\in A}w_{i}]$\@.
An equivalent way of looking at running the greedy algorithm on the
random set $S$ is the following:
\begin{enumerate}
\item $A=\emptyset$
\item visit all agents $i\in N$ in decreasing order of weights:

\begin{enumerate}
\item if $A+i\in\mathcal{I}$, we check if $i$ is in $S$, and add $i$
into $A$ if yes. 
\item if $A+i\notin\mathcal{I}$, we ignore $i$. 
\end{enumerate}
\item output $A$ 
\end{enumerate}
Let random set $U$ be the set of agents that are ignored by greedy.
Consider the quantity $Q=E_{S\sim\mathbf{q}}[\sum_{i\in A}w_{i}+\sum_{i\in U}\tilde{q}_{i}w_{i}]$.
For every agent $i$, if she is checked by greedy, she contributes
$q_{i}w_{i}$ to the $Q$. (with probability $q_{i}$, $i$ is in
$S$, and we get weight $w_{i}$) On the other hand, if she is ignored,
she contributes $\tilde{q}_{i}w_{i}$ to $Q$. Therefore, \[
Q=E_{S\sim\mathbf{q}}[\sum_{i\in A}w_{i}+\sum_{i\in U}\tilde{q}_{i}w_{i}]\geq\sum_{i\in N}\tilde{q}_{i}w_{i}=E_{S\sim\mathcal{D}}[w^{*}(S)].\]

Next we show that $w(A)\geq\frac{1}{p}\sum_{i\in U}\tilde{q}_{i}w_{i}$,
and our theorem would follow as: \begin{eqnarray*}
E_{S\sim\mathbf{q}}[\sum_{i\in greedy(S)}w_{i}] & = & E_{S\sim\mathbf{q}}[\sum_{i\in A}w_{i}]\\
 & \geq & \frac{1}{p+1}E_{S\sim\mathcal{D}}[w^{*}(S)].\end{eqnarray*}

Let $A$ contain $i_{1},i_{2},\ldots,i_{l}$ in the order of inclusion
into $A$ by greedy. Partition $U$ into $B_{j}$'s for $j=1,\dots,l$,
where $B_{j}$ is the set of agents ignored by greedy after $i_{1},\dots,i_{j}$
have been added into $A$. Therefore $w_{i}\leq w_{i_{j}}$ for $i\in B_{j}$.
Consider the set $\{i_{1},\dots,i_{j}\}\cup B_{1}\cup\dots\cup B_{j}$.
At any time step, greedy's solution set is always a maximal feasible
subset of the agents visited so far. Therefore $\{i_{1},\dots,i_{j}\}$
is a base of $\{i_{1},\dots,i_{j}\}\cup B_{1}\cup\dots\cup B_{j}$.
By the definition of $p$-independence, the maximal base of $\{i_{1},\dots,i_{j}\}\cup B_{1}\cup\dots\cup B_{j}$
has size at most $p\cdot j$, and it follows that $\sum_{i\in B_{1}\cup\dots\cup B_{j}}\tilde{q}_{i}\leq p\cdot j$
.

Now our claim $\sum_{i\in A}w_{i}\geq\frac{1}{p}\sum_{i\in U}\tilde{q}_{i}w_{i}$
follows from the following inequalities: (let $w_{i_{l+1}}=0$)

\begin{eqnarray*}
\sum_{i\in U}\tilde{q}_{i}w_{i} & = & \sum_{1\leq j\leq l}\sum_{i\in B_{j}}\tilde{q}_{i}w_{i}\\
 & \leq & \sum_{1\leq j\leq l}\sum_{i\in B_{j}}\tilde{q}_{i}w_{i_{j}}\\
 & = & \sum_{1\leq j\leq l}\sum_{i\in B_{1}\cup\dots\cup B_{j}}\tilde{q}_{i}(w_{i_{j}}-w_{i_{j+1}})\\
 & \leq & \sum_{1\leq j\leq l}p\cdot j\cdot(w_{i_{j}}-w_{i_{j+1}})\\
 & = & p\cdot\sum_{1\leq j\leq l}w_{i_{j}}=p\cdot\sum_{i\in A}w_{i}.\end{eqnarray*}

\end{proof}

\subsection{Proof of Proposition \ref{pro:p_ind_lb}}

\begin{proof}

To define the set system $(N,\mathcal{I})$, let $Y$ be the set of
all strings $a_{1}a_{2}\dots a_{n}$ of length $n$ over the alphabet
$\{1,\dots,n\}$. For every $i\in\{1,2,\dots,n\}$ and $b\in\{1,\dots,n\}$,
we denote by $[a_{i}=b]$ the {}``miniset'' that contains all strings
from $Y$ with the $i$-th letter $a_{i}$ being $b$. Then $N$ is
the set of all such minisets. To define the feasible subsets $\mathcal{I}$,
a subset $S$ of minsets from $N$ is feasible if and only if no two
minisets in $S$ intersect. Note that two different minisets $[a_{i}=b]$
and $[a_{i'}=b']$ intersect if and only if $i\neq i'$. It is easy
to verify that this set system is $n$-independent. Finally, we assign
unit weights to every miniset.

We choose a random subset $S$ of $N$ in two ways. In the dependent
case, an index $i$ from $\{1,\dots,n\}$ is chosen at random, and
$S$ contains the miniset $[a_{i}=b]$ for all $b\in\{1,\dots,n\}$.
Clearly all such $S$'s are feasible, and the rank function has expected
value $n$.

In the independent case, for all $i,b$, we include every miniset
$[a_{i}=b]$ in $S$ with probability $1/n$ independently. For all
$i$, let $X_{i}$ be the number of minisets in $S$ that have the
form $[a_{i}=b]$ for some $b$. Then the rank function is equal to
$\max_{i}X_{i}$. To give a rough estimate of $E[\max_{i}X_{i}]$,
note that for all $i$, \begin{eqnarray*}
 &  & Pr[X_{i}\geq\frac{1}{2}\log n]\\
 & = & \sum_{k=\frac{1}{2}\log n}^{n}{n \choose k}\left(\frac{1}{n}\right)^{k}\left(1-\frac{1}{n}\right)^{n-k}\\
 & \leq & \sum_{k=\frac{1}{2}\log n}^{n}\left(\frac{n\cdot e}{k}\right)^{k}\frac{1}{n^{k}}\leq n\cdot\left(\frac{e}{\frac{1}{2}\log n}\right)^{\frac{1}{2}\log n}\\
 & = & \frac{n}{2^{\Omega(\log n\cdot\log\log n)}}.\end{eqnarray*}
 Therefore for sufficiently large $n$, $Pr[\max_{i}X_{i}\geq\frac{1}{2}\log n]\leq1-\left(1-\frac{n}{2^{\Omega(\log n\cdot\log\log n)}}\right)^{n}\leq\frac{1}{n}$,
and hence $E[\max_{i}X_{i}]\leq Pr[\max_{i}X_{i}\geq\frac{1}{2}\log n]\cdot n+\frac{1}{2}\log n\leq\log n$.
It follows that the correlation gap is at least $n/\log n$ for sufficiently
large $n$. 

\end{proof}
\end{document}